\documentclass[12pt]{article}

\usepackage[english]{babel}
\usepackage[utf8]{inputenc}
\usepackage[T1]{fontenc}

\usepackage{float}% to force figure placement with H
\usepackage{multirow,multicol}

\usepackage{bm}
\usepackage{amsmath}
\usepackage{amsfonts}
\usepackage{amsthm} %% pour les theoremes

\usepackage[]{algorithm2e} %% algos

\usepackage{numprint}

\usepackage{mathrsfs} %Calligraphy

\usepackage{mdframed} %% Box frames
\newmdenv{allfour}
\newmdenv[leftline=false,rightline=false]{topbot}
\newmdenv[topline=false,rightline=false]{leftbot}
\newmdenv[topline=false,leftline=false]{rightbot}
\newmdenv[topline=false,rightline=false,leftline=false]{bottom}

\usepackage[justification=centering]{caption}
\usepackage{subcaption} %to have subfigures available

\usepackage[pdftex]{graphicx}

\DeclareGraphicsExtensions{.jpg,.pdf,.png}

\usepackage[pdftex,colorlinks=true,linkcolor=blue,citecolor=blue,urlcolor=blue]{hyperref}

\usepackage{anysize}
\marginsize{22mm}{14mm}{12mm}{25mm}

\usepackage{icomma}

\usepackage{xcolor, colortbl}

\usepackage{authblk}

\newcommand{\tr}[1]{#1^{T}}

\newcommand{\trace}{\mathop{\text{Trace}}}

\newtheorem{theorem}{Theorem}[section]

\newtheorem{definition}{Definition}[section]

\newtheorem{prop}{Proposition}[section]

\newenvironment{nscenter}
 {\parskip=0pt\par\nopagebreak\centering}
 {\par\noindent\ignorespacesafterend}

\definecolor{col8}{gray}{1}
\definecolor{col7}{gray}{0.9}
\definecolor{col6}{gray}{0.8}
\definecolor{col5}{gray}{0.7}
\definecolor{col4}{gray}{0.6}
\definecolor{col3}{gray}{0.5}
\definecolor{col2}{gray}{0.4}
\definecolor{col1}{gray}{0.35}

\providecommand{\keywords}[1]
{
  \small    
  \textbf{\textit{Keywords---}} #1
}

\begin{document}

\setlength{\abovedisplayskip}{1ex}
\setlength{\belowdisplayskip}{1ex}

%% Voilà mes légendes de figures comme je les aime:
\makeatletter
\def\figurename{{\protect\sc \protect\small\bfseries Fig.}}
\def\f@ffrench{\protect\figurename\space{\protect\small\bf \thefigure}\space}
\let\fnum@figure\f@ffrench%
\let\captionORI\caption
\def\caption#1{\captionORI{\rm\small #1}}
\makeatother

%%%%%%%%%%%%%%%%%%%%%%%%%%%%%%%%%%%%%%%%%%%%%%%%%%%%%%%%%% Couverture:
\title{Efficient simulation of Gaussian Markov random fields by Chebyshev polynomial approximation}

\author[1,2]{Mike PEREIRA%
  \thanks{Contact: \texttt{mike.pereira@mines-paristech.fr}}}  
\author[1]{Nicolas DESASSIS%
  \thanks{Contact: \texttt{nicolas.desassis@mines-paristech.fr}}}
  
\affil[1]{Geostatistics team, Geosciences, MINES ParisTech, PSL Research University\\
   Fontainebleau, France}
\affil[2]{Estimages\\
   Paris, France}

\date{\today}

\maketitle

\begin{abstract}
This paper presents an algorithm to simulate Gaussian random vectors  whose precision matrix can be expressed as a polynomial of a sparse matrix. This situation arises in particular when simulating Gaussian Markov random fields obtained by the finite elements discretization of the solutions of some stochastic partial derivative equations. The proposed algorithm uses a Chebyshev polynomial approximation to compute simulated vectors with a linear complexity. This method is asymptotically exact as the approximation order grows. Criteria based on tests of the statistical properties of the produced vectors are derived to determine minimal orders of approximation.
\end{abstract}


\vspace{9ex}

%%%%%%%%%%%%%%%%%%%%%%%%%%%%%%%%%%%%%%%%%%%%%%%%%%%%%%%%%%

\section{Introduction}

Gaussian random fields (GRF) are widely used to model spatially correlated data  in environmental and earth sciences \cite{chiles1999geost,lantuejoul2013geostatistical,wackernagel2013multivariate}. The stochastic simulation of such fields (also called geostatistical simulation) is a common process in risk analysis \cite{chiles1999geost}. Indeed, each simulation is seen as an alternate but plausible version of the reality. Spatial uncertainty can then be assessed in problems where the variables of interest are partially observed through comparisons over a set of simulations.
There are two main classes of simulation algorithms. Exact algorithms aim at reproducing exactly the statistical properties of a targeted model. They include methods based on the factorization of covariance matrices \cite{davis1987sim} or on the spectral properties of random fields \cite{Pardo1993,dietrich1993fast}. On the other hand, approximate algorithms generate simulations with a nearly multi-Gaussian spatial distribution or with approximated covariance properties. They were introduced to tackle the large-scale limitations of exact algorithms. Examples of such algorithms include the turning bands \cite{matheron1973intrinsic,emery2006tbsim} and the continuous spectral methods \cite{shinozuka1972digital}, and also the sequential Gaussian simulation algorithm \cite{deutsch1998}. 

Continuous Markov random fields are particularly suited models for geostatistical simulations thanks to the computational efficiency they provide. Precisely, the sparsity of the precision matrices of their discretization allows fast computations of samples (and likelihood) \cite{rue2005gaussian}. When stationary and isotropic, these random fields have a spectral density, that is the Fourier transform of the covariance function, of the form $f(\omega)=1/\mathrm{P}(\Vert \omega\Vert^2)$ where $\mathrm{P}$ is a real  strictly positive polynomial on $\mathbb{R}_+$ \cite{rozanov1977}. Equivalently, they can be seen as solutions of the stochastic partial derivative equation (SPDE) defined as \cite{rozanov1977, lang2011, simspon2016}: 
\begin{equation}
\mathrm{P}(-\Delta)^{1/2}Z=\mathcal{W}
\label{spde_gen}
\end{equation}
where $\mathcal{W}$ is a Gaussian white noise and $\mathrm{P}(-\Delta)^{1/2}$ is the differential operator defined as:
$$\mathrm{P}(-\Delta)^{1/2}[.]=\mathscr{F}^{-1}\left[w\mapsto \sqrt{\mathrm{P}(\Vert \omega \Vert ^2)}\mathscr{F}[.](\omega)\right]$$
where $\mathscr{F}$ denotes the Fourier transform operator. 

For instance, following the results from Whittle \cite{whittle1954stationary}, Lindgren et al. \cite{lindgren2011explicit} considers stationary solutions of the SPDE:
\begin{equation}
(\kappa^2-\Delta)^{\alpha/2}Z=\tau\mathcal{W}
\label{spdeMatern}
\end{equation}
with  $\kappa>0$, $\tau >0$ and $\alpha$ an integer greater than half the dimension of the space, to characterize GRFs with Matérn covariance (or Matérn fields). They even use this result to extend isotropic Matérn fields to manifolds, and to non-stationary and even oscillating formulations \cite{lindgren2011explicit}.

SPDE \eqref{spde_gen} can be numerically solved using the finite element method. In that case, it is solved on a triangulated domain, and a finite element representation of the solution is built as: 
\begin{equation*}
Z(x)=\sum\limits_{i} z_i\psi_i(x)
\end{equation*}
for finite and deterministic basis functions $\lbrace \psi_i\rbrace$ and Gaussian weights $\lbrace z_i \rbrace$. Simulating a solution is then equivalent to simply simulate the Gaussian weights $\lbrace z_i \rbrace$. In particular, the precision matrix of these weights can be specified using weak formulations of the SPDE, and has the form:
\begin{equation}
\bm{Q}=\bm{D}\mathrm{P}(\bm S) \bm{D}
\label{Q_gen}
\end{equation} 
where $\bm D$ is a diagonal matrix with strictly positive entries and $\bm S$ is a real, symmetric and positive semi-definite matrix. In particular, when piecewise linear basis functions are considered, $\bm S$ is a very \textit{sparse} matrix, whose non-zero entries correspond to adjacent nodes in the triangulation. 

Given that the precision matrix is known, the simulation of solutions is generally performed by matrix factorisation methods involving the Cholesky decomposition. Even if the sparsity inherited from the Markovian properties of the field reduces the complexity of an otherwise too expensive factorisation \cite{davis2006direct}, computation and storage problems still arise for large simulation domains or when the dimension of the space increases \cite{simpson2008fast}. Methods based on  iterative techniques for solving sparse linear systems were proposed to avoid the Cholesky decomposition \cite{simpson2008fast,Aune2013}. They rely on the fact that the product of $\bm Q$ and a vector is assumed to be inexpensive, and on finding preconditioners to reduce the number of iterations needed for the algorithm to converge \cite{simpson2013scalable}.

This article introduces instead a computationally efficient algorithm to simulate any Gaussian random vectors whose precision matrix can be expressed as \eqref{Q_gen}. This algorithm is based on the construction of a polynomial approximation of a factorisation of $\bm Q$. It then relies on matrix-vector products between (a matrix as sparse as) $\bm S$ and vectors. It can produce simulations of vectors with a linear complexity, proportional to the number of non-zero entries of $\bm S$.

The simulation algorithm presented in this article is equivalent to a filtering technique used in Graph signal processing (GSP) \cite{hammond2011wavelets}. GSP is an emerging field focusing on developing tools  to process complex data that are embedded on a graph, i.e. a structure composed of a set of objects, called vertices, and pairwise relationships between them, the edges \cite{bondy1976graph}. Such data arise naturally in applications such as social, energy, transportation and neural networks. They are modelled as variables indexed by the vertices of the graph, named graph signals. Generalizations of classical signal processing notions and tools, such as the Fourier transform, filtering and translation operators are then used to study these signals \cite{shuman2013emerging}. 

The outline of the article is as follows. In Section \ref{sec:sim}, methods for the simulation of Gaussian random vectors with known precision matrix are reviewed. In Section \ref{sec:pol}, the main idea behind the proposed algorithm is introduced and attention is devoted to the polynomial approximation it is based on. In Section \ref{sec:alg}, the overall workflow of the algorithm is presented, and its complexity and induced error are calculated. Then the framework of statistical tests is used to assess whether the vectors produced by the algorithm respect their targeted distribution, and criteria on the minimal order of approximation are deduced. Finally, in Section \ref{sec:app}, examples of application of the algorithm are presented, highlighting the great adaptability of the algorithm for the simulation of Matérn fields and their generalizations.

\section{Simulation of Gaussian random vectors}
\label{sec:sim}

The aim is to simulate a zero-mean Gaussian random vector (GRV) whose precision matrix $\bm Q$ is given by:
\begin{equation}
\bm{Q}=\bm{D}\mathrm{P}(\bm S) \bm{D}=\bm{D}\bigg(\sum\limits_{l=0}^{L}b_l \bm S^l \bigg)\bm{D}
\label{qpolmat}
\end{equation} 
where $\mathrm{P} : x\mapsto \sum_{l=0}^{L}b_l x^l$ is a strictly positive polynomial function on $\mathbb{R}_+$; $\bm S$ is a real, sparse, symmetric and positive semi-definite matrix; and $\bm D$ is an invertible diagonal matrix.
\subsection{Simulation by matrix factorisation}
\label{subsec:matfac}
A non-conditional simulation of a zero-mean GRV $\bm z$ with known precision matrix $\bm Q$ can be obtained through:
\begin{equation}
\bm z = \bm L \bm\varepsilon
\label{sim_mat}
\end{equation}
where $\bm\varepsilon$ is a vector with independent zero-mean, unit variance and normally distributed random components and $\bm L$ is a matrix such that \cite{Gentle2009}:
\begin{equation}
\bm L \bm L^T=\bm Q^{-1}
\label{decomp}
\end{equation}
The most widely used candidate for such a matrix $\bm L$ is the Cholesky decomposition of $\bm Q^{-1}$ \cite{horn1990matrix}. However, in the considered setting, only the precision matrix $\bm{Q}$ is known and not its inverse. Therefore, the simulation process can be performed in two steps:
\vspace{0.5ex}

\begin{topbot}[innertopmargin=2ex,innerbottommargin=1ex]
\vspace{-1ex}
\textbf{Workflow}: Simulation of a random vector using Cholesky decomposition

\vspace{-2ex}
\noindent\makebox[\linewidth]{\rule{\textwidth}{0.5pt}}\\\textbf{Require}:  A precision matrix $\bm Q$. A vector of independent standard Gaussian values $\bm\varepsilon$.\\
\textbf{Output}: A simulated vector $\bm z$ with precision matrix $\bm{Q}$.
\vspace{-1ex}
\begin{enumerate}
\item Compute $\bm Q_{\text{chol}}$ the Cholesky decomposition of the precision matrix $\bm Q$.
\item Compute the simulated vector $\bm z$ as the solution of the following linear system: 
$$\tr{\bm Q_{\text{chol}}}\bm z=\bm\varepsilon$$
\end{enumerate}
\end{topbot}

\vspace{1ex}
Two performance issues arise from this workflow. First, the computation of the Cholesky decomposition of  $\bm Q$ is intractable for large problems or when the matrix is not sparse enough \cite{simpson2008fast}. Then, once computed, this decomposition must be stored, and is used to solve a linear system. Both these tasks grow more expansive as the size or the filling of $\bm Q_{\text{chol}}$ increases. The idea behind the algorithm presented in this article is to find another candidate for $\bm L$ that would take advantage of the fact that the precision matrix has the form \eqref{qpolmat}. 

\subsection{Simulation by eigendecomposition}
%\subsubsection*{Formulation}
The matrix $\bm S$ being real and symmetric, it is diagonalizable with non-negative eigenvalues $\lambda_1, \dots, \lambda_n$ and eigenvectors that form an orthonormal basis of $\mathbb{R}^n$ (with $n$ the size of the matrix $\bm S$). Therefore there exists a matrix $\bm V$ satisfying $\bm V^{-1}=\bm V^T$ and: 
\begin{equation*}
\bm S=\bm V \begin{pmatrix} \lambda_1 & & \\  & \ddots &  \\  &  &  \lambda_n \end{pmatrix} \bm V^{-1}
\end{equation*}
It can be shown that for any real polynomial $\mathrm{P}$, $\mathrm{P}(\bm S):=\sum_{l=0}^{L}b_l \bm S^l$ is also a real symmetric matrix, and is diagonalizable in the same eigenbasis as $\bm S$. In particular, the eigenvalues of $\mathrm{P}(\bm S)$ are $\mathrm{P}(\lambda_1), \dots, \mathrm{P}(\lambda_n)$.  

Let's then denote $\mathrm{P}(\bm S)^{-1/2}$ the matrix defined for strictly positive polynomials $\mathrm{P}$ by:
\begin{equation}
\mathrm{P}(\bm S)^{-1/2}:=\bm V \left(\begin{smallmatrix} 1/\sqrt{\mathrm{P}(\lambda_1)} & & \\  & \ddots &  \\  &  &  1/\sqrt{\mathrm{P}(\lambda_n)} \end{smallmatrix} \right) \bm V^{-1}
\label{sq_def}
\end{equation}
Given that this matrix is symmetric, $\bm L = \bm D^{-1} \mathrm{P}(\bm S)^{-1/2}$ satisfies \eqref{decomp}. So using \eqref{sim_mat}, a field $\bm z$ with precision matrix $\bm Q$ can be generated through:
\begin{equation}
\bm z = \bm D^{-1} \mathrm{P}(\bm S)^{-1/2}\bm\varepsilon
\label{L_expr}
\end{equation} 

A direct way to compute the matrix $\mathrm{P}(\bm S)^{-1/2}$ is through \eqref{sq_def} which supposes to diagonalize the matrix $\bm S$, and store its eigenvalues $\lambda_1, \dots, \lambda_n$ and eigenvectors $\bm V$. The cost associated to this approach is generally prohibitive as it requires $\mathcal{O}(n^3)$ operations and a storage size of $\mathcal{O}(n^2)$. 

However, a particular case when this is feasible is worth noticing. In \cite{rue2005gaussian}, Rue and Held show that the precision matrix of a stationary Gaussian Markov random field defined on a torus (i.e. a regular lattice with  cyclic boundary conditions) is block-circulant, with circulant blocks. They deduce that the eigenvalues and the eigenvectors of the precision matrix can be computed using the discrete Fourier Transform (DFT), and therefore without requiring a matrix diagonalization. They then sample from their Gaussian Markov random field using \eqref{L_expr}, where, following the notations of this section, $P(X )=X$, $\bm D$ is the identity matrix, and so $\bm S = \bm Q$. They just replace the product between the $\bm V$ (resp. $\bm V^{-1}$) and a vector by the DFT (resp. inverse DFT) of this vector.

%%%%%%%%%%%%%%%%%%%%%%%%%%%%%%%%%%%%%%%%%%%%%%%%%%%%%%%%%%

\section{Polynomial approximation}
\label{sec:pol}

In the general case, the eigendecomposition of $\bm S$ is inevitable if \eqref{L_expr} is used to simulate the GRV. To avoid this expensive operation, the idea is rather to compute a matrix polynomial approximation $\mathrm{P}_{\mbox{\scriptsize -1/2}}(\bm S):=\sum_{k=0}^{K}\alpha_k \bm S^k$ of $\mathrm{P}(\bm S)^{-1/2}$. Indeed, both matrices can be decomposed in the same eigenbasis $\bm V$ as: 
\begin{equation*}
\mathrm{P}_{\mbox{\scriptsize -1/2}}(\bm S)=\bm V \left(\begin{smallmatrix} \mathrm{P}_{\mbox{\tiny -1/2}}(\lambda_1) & & \\  & \ddots &  \\  &  &  \mathrm{P}_{\mbox{\tiny -1/2}}(\lambda_n) \end{smallmatrix} \right) \bm V^T \quad \text{and}\quad \mathrm{P}(\bm S)^{-1/2}:=\bm V \left(\begin{smallmatrix} 1/\sqrt{\mathrm{P}(\lambda_1)} & & \\  & \ddots &  \\  &  &  1/\sqrt{\mathrm{P}(\lambda_n)} \end{smallmatrix} \right) \bm V^T
\end{equation*}
Consequently, to approximate $ \mathrm{P}(\bm S)^{-1/2}$ by $\mathrm{P}_{\mbox{\scriptsize -1/2}}(\bm S)$, the polynomial $\mathrm{P}_{\mbox{\scriptsize -1/2}}$ must satisfy: 
\begin{equation}
\forall i\in [\![1,n]\!], \quad \mathrm{P}_{\mbox{\scriptsize -1/2}}(\lambda_i) \approx 1/\sqrt{\mathrm{P}(\lambda_i)}
\label{approx_pol}
\end{equation} 
In that case, using  \eqref{L_expr}, a field $\bm z$ with precision matrix approximately equal to $\bm Q$ can be simulated via the formula:
\begin{equation}
\bm z=\bm D^{-1} \mathrm{P}_{\mbox{\scriptsize -1/2}}(\bm S) \bm\varepsilon = \bm D^{-1}\sum_{k=0}^{K}\alpha_k \bm S^k \bm\varepsilon
\label{sim_pol}
\end{equation}
Once the $\lbrace \alpha_k\rbrace$ are known, computing the simulated field using \eqref{sim_pol} can be done using an iterative algorithm that only requires matrix-vector products involving the sparse matrix $\bm S$. 

To define the expression of a polynomial satisfying \eqref{approx_pol}, it is sufficient to solve the following problem: \textit{given an interval $[a, b]$ containing all the eigenvalues of $\bm S$, find a polynomial $\mathrm{P}_{\mbox{\scriptsize -1/2}}$ that approximates the (continuous) function $x\mapsto 1/\sqrt{\mathrm{P}(x)}$ over $[a, b]$}. 
Such an interval $[a, b]$ can be obtained \textit{without having to diagonalize $\bm S$}. Examples of such intervals are provided in Appendix \ref{appen:eig_bound}. Using these results and the fact that $\bm S$ is positive semi-definite, the following interval is considered in the applications presented in this paper: 
\begin{equation}
[a, b]=\big[0, \max\limits_{i\in [\![1,n]\!]} {\sum\limits_{j\in [\![1,n]\!]} |S_{ij}|}\big]
\label{a_b_expr}
\end{equation}

The approximation of the function $1/\sqrt{\mathrm{P}}$ over the interval $[a, b]$ is carried out using Chebyshev polynomials \cite{mason2002chebyshev,press2007numerical}. This family  $(T_k)_{k\in\mathbb{N}}$ of polynomials is the sequence of polynomials defined over $[-1,1]$ by:
$$\forall \theta\in\mathbb{R}, \quad T_k(\cos\theta)=\cos(k\theta)$$
or equivalently via the recurrence relation:
\begin{equation*}
T_0(x) =1, \quad
T_1(x) =x, \quad
T_{k+1}(x) =2xT_k(x)-T_{k-1}(x) \quad (k\geq 1)
\end{equation*}
A graphical representation of these polynomials is provided in Figure \ref{chebpolfig}. Notice that they can be  generalized to arbitrary intervals $[a, b]\subset\mathbb{R}$ via a simple change of variable:
\begin{equation}
y \in [a, b] \mapsto x=\frac{2y-b-a}{b-a} \in [-1, 1]
\label{changeVar}
\end{equation}

\begin{figure}[t]
\centering
\includegraphics[scale=0.6]{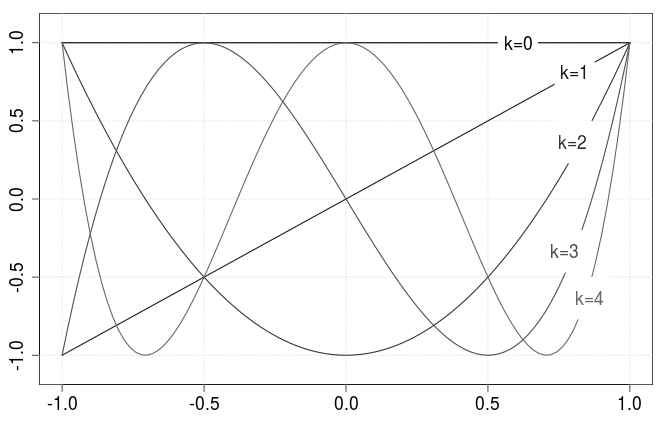}
\vspace{-2ex}
\caption{First 5 Chebyshev polynomials over $[-1, 1]$.}
\label{chebpolfig}
\end{figure}

The choice of Chebyshev polynomials has several perks. For the sake of simplicity, they are listed in the case where a function $f$ defined on $[-1,1]$ has to be approximated. 
\begin{itemize}
\item \textit{Convergence} \cite{mason2002chebyshev}: If $f$ is Lipschitz-continuous over $[-1, 1]$ (which is the case in the applications presented in this article), its Chebyshev series, defined as:
\begin{equation}
\forall x\in[-1, 1], \quad S(x)=\frac{1}{2}c_0T_0\left(x\right) +\sum\limits_{k=1}^{\infty} c_k T_k\left(x\right)
\end{equation}
where $\forall k\in\mathbb{N}$: 
\begin{equation}
\label{coefCheb}
c_k = \frac{2}{\pi}\int_{-1}^{1} f(x)T_k(x)\frac{1}{\sqrt{1-x^2}}dx
\end{equation}
is uniformly convergent on $[-1, 1]$. A similar result can be obtained for continuous functions when the Cesàro sums of their Chebyshev series are considered instead. 

\item \textit{Near minimax property} \cite{mason2002chebyshev,press2007numerical}: Suppose that $f$ is continuous. The minimax polynomial of degree $K$ of $f$ is the polynomial $p^*_K$ defined as:
$$p^*_K=\mathop{\text{argmin}}\limits_{p\in\mathscr{P}_K}\Vert f - p \Vert_\infty$$
where $\mathscr{P}_K$ is the set of all polynomials of degree $\le K$ and $\Vert . \Vert_\infty$ denotes the uniform norm (over $[-1, 1]$). It is generally very difficult to compute. However, the truncated Chebyshev series of degree $K$, denoted $S_K$, is a good approximation of $p^*_K$ in the sense that $\Vert f - S_K \Vert_\infty$ is close to $\Vert f - p^*_K \Vert_\infty$.

\item \textit{Fast computation} \cite{press2007numerical}: A change of variable in \eqref{coefCheb} gives:
\begin{equation}
c_k =\frac{2}{\pi}\int_{0}^{\pi} f\left(\cos\theta\right)\cos(k\theta)d\theta \approx \sum\limits_{j=0}^J f(\cos(j\frac{\pi}{J}))\cos(k j\frac{\pi}{J})
\label{dft}
\end{equation}
This last sum is the expression of the real part of the discrete Fourier Transform of the vector $\left(f(1),\dots,f(\cos(j\frac{\pi}{J})),\dots,f(-1)\right)^T$ for discretization order $J\in\mathbb{N}^*$. Hence the coefficients of the Chebyshev series of a function can be numerically computed using the Fast Fourier Transform algorithm, known for its speed and accuracy \cite{brigham1988fast}. 
\end{itemize}

%%%%%%%%%%%%%%%%%%%%%%%%%%%%%%%%%%%%%%%%%%%%%%%%%%%%%%%%%%

\section{Simulation algorithm}
\label{sec:alg}

In this section, the workflow of the simulation algorithm is presented, then its complexity and the induced error are derived. Finally, criteria on the choice of the approximation order are given.

\begin{figure}[H]
\begin{topbot}[innertopmargin=2ex,innerbottommargin=1ex]
\vspace{-1ex}
\textbf{Workflow}: Simulation of a random vector using Chebyshev approximation

\vspace{-2ex}
\noindent\makebox[\linewidth]{\rule{\textwidth}{0.5pt}}\\
\textbf{Require}:  A positive polynomial $\mathrm{P}$, a real symmetric positive semi-definite $n\times n$ matrix $\bm S$ and an invertible diagonal matrix $\bm D$ of size $n$. An order of approximation $K\in\mathbb{N}$. A vector of $n$ independent standard Gaussian components $\bm\varepsilon$.\\
\textbf{Output}: A vector $\bm z$ with precision matrix (approximately equal to) $\bm{Q}= \bm{D}\mathrm{P}(\bm S) \bm{D}$
\vspace{-1ex}
\begin{enumerate}
\item Find an interval $[a, b]$ containing all the eigenvalues of $\bm S$ (for instance \eqref{a_b_expr}).
\item Compute a polynomial approximation $\mathrm{P}_{\mbox{\scriptsize -1/2}}$ of the function $x \mapsto {1}/{\sqrt{\mathrm{P}(x)}}$ over $[a, b]$, by truncating its shifted Chebyshev series at order $K$:
\begin{equation*}
\mathrm{P}_{\mbox{\scriptsize -1/2}}(x)=\frac{1}{2}c_0T^{[a, b]}_0(x) +\sum\limits_{k=1}^{K} c_k T^{[a, b]}_k(x),\quad T^{[a, b]}_k(x) :=T_k\left(\frac{2}{b-a}x-\frac{b+a}{b-a}\right)
\end{equation*}
The coefficients $(c_k)_{k\in[\![0,K]\!]}$ are computed by Fast Fourier Transform (using the changes of variables \eqref{changeVar} in \eqref{dft}).
\item Compute the product $\bm u = \mathrm{P}_{\mbox{\scriptsize -1/2}}(\bm S)\bm\varepsilon$ using the recurrence relation satisfied by the Chebyshev polynomials.\\
\vspace{-3ex}
%\begin{align*}
%& T^{[a, b]}_1(\bm S) :=\frac{2}{b-a}\bm S -\frac{b+a}{b-a}\bm I; & \quad k=0;\\
%& \bm u_0 = \bm\varepsilon;  \quad \bm u = \frac{1}{2}c_0\bm u_0; &\quad k\leftarrow k+1; \\
%& \bm u_1 =T^{[a, b]}_1(\bm S)\bm u_0; \quad \bm u \leftarrow \bm u+ c_1\bm u_1; &\quad k\leftarrow k+1; \\
%&{\small \textbf{While}}(k \le K)\lbrace \\
%& \quad\quad \bm u_{k} =2T^{[a, b]}_1(\bm S)\bm u_{k-1}-\bm u_{k-2}; \quad 
%\bm u \leftarrow \bm u + c_{k}\bm u_{k}; &\quad k\leftarrow k+1; \\
%& \rbrace \\
%&{\small \textbf{Return }} \bm u
%\end{align*}
\begin{align*}
& \alpha :=\frac{2}{b-a}; \quad  \beta:=\frac{b+a}{b-a}; & \quad k=0;\\
& \bm u^{(-2)} = \bm\varepsilon;  \quad \bm u = \frac{1}{2}c_0\bm u^{(-2)}; &\quad k\leftarrow k+1; \\
& \bm u^{(-1)} =\alpha\bm S\bm\varepsilon-\beta\bm\varepsilon; \quad \bm u \leftarrow \bm u+ c_1\bm u^{(-1)}; &\quad k\leftarrow k+1; \\
&{\small \textbf{While}}(k \le K)\lbrace \\
& \quad\quad \bm u^{(0)} =\alpha\bm S\bm u^{(-1)}-\beta\bm u^{(-1)}-\bm u^{(-2)}; \quad 
\bm u \leftarrow \bm u + c_{k}\bm u^{(0)}; \\
& \quad\quad \bm u^{(-2)}\leftarrow \bm u^{(-1)}; \quad \bm u^{(-1)}\leftarrow\bm u^{(0)};  &\quad k\leftarrow k+1; \\
& \rbrace \\
&{\small \textbf{Return }} \bm u
\end{align*}
\item The simulated field is given by: $\bm z = \bm D^{-1} \bm u$
\end{enumerate}
\end{topbot}
\end{figure}
\vspace{-20pt}

\subsection{Complexity of the algorithm}

The complexity of the simulation algorithm can be explicitly calculated. Denote $n_{nz}$ the number of non-zero entries of $\bm S$ and  $m_{nz}$ the mean number of non-zero entries of a row of $\bm S$: $n_{nz}=m_{nz}\times n$. \\
Denote $K$ the order of the Chebyshev approximation. 
The cost associated with each step (ignoring additions and multiplications by non-stored zeros) is described as follows:
\begin{itemize}
\item Step 1 requires $\mathcal{O}(n_{nz})$ operations using \eqref{a_b_expr} to compute the interval $[a, b]$.
\item Step 2 requires to compute the Fast Fourier Transform of a vector of length $K$. The cost of this operation is $\mathcal{O}(K\log K)$. 
\item Step 3 requires to:
\begin{itemize}
\item compute $K$ products of matrix $\bm S$ and vectors $\rightarrow$ $\mathcal{O}(Kn_{nz})$ operations
\item  $2(K-1)+1$ subtractions of vectors, $3(K-1)+1$ multiplications of a vector by a scalar value and $(K-1)$ additions of vectors $\rightarrow$ $(K-1)n$ operations.
\end{itemize}
\item Step 4 requires $n$ operations (product of a diagonal matrix and a vector).
\end{itemize} 
Therefore, the overall cost of the simulation algorithm is $\mathcal{O}(K n_{nz})=\mathcal{O}(Km_{nz}n)$ operations.\\
And regarding the storage needs, aside from $\bm S$, $\bm D$ and $\bm \varepsilon$ which are assumed to be known (and therefore stored), the algorithm only needs to store $4$ additional vectors of size $n$ ($\bm u$, $\bm u^{(0)}$, $\bm u^{(-1)}$ and $\bm u^{(-2)}$).

%%%%%%%%%%%%%%%%%%%%%%%%%%%%%%%%%%%%%%%%%%%%%%%%%%%%%%%%%%

\subsection{Error quantification}
\label{sec:err}

The main idea of the simulation algorithm presented in this article is to replace the matrix $\bm D^{-1}\mathrm{P}(\bm S)^{-1/2}$ in relation \eqref{L_expr} by an efficient polynomial approximation, namely the matrix $\bm D^{-1}\mathrm{P}_{\mbox{\scriptsize -1/2}}(\bm S)$. The approximation error $\epsilon_{\text{approx}}$ between these matrices can be measured by :
\begin{equation*}
\epsilon_{\text{approx}}:=\Vert\bm D^{-1}\mathrm{P}(\bm S)^{-1/2}-\bm D^{-1}\mathrm{P}_{\mbox{\scriptsize -1/2}}(\bm S)\Vert_{\infty}
\end{equation*}
where $\Vert . \Vert_{\infty}$ denotes the matrix max norm, defined by $\Vert\bm A\Vert_{\infty}:=\max\limits_{i,j} |A_{ij}|$. $\bm D$ being a diagonal matrix, 
\begin{equation*}
\epsilon_{\text{approx}}\leq \Vert\bm D^{-1}\Vert_{\infty}\Vert\mathrm{P}(\bm S)^{-1/2}-\mathrm{P}_{\mbox{\scriptsize -1/2}}(\bm S)\Vert_{\infty}\leq \Vert\bm D^{-1}\Vert_{\infty}\Vert\mathrm{P}(\bm S)^{-1/2}-\mathrm{P}_{\mbox{\scriptsize -1/2}}(\bm S)\Vert_{2}
\end{equation*}
where $\Vert . \Vert_{2}$ denotes the Froebenius norm, defined by $\Vert\bm A\Vert_{2}:=\sqrt{\text{Trace}(\bm A \bm A^T)}$. Therefore,
\begin{equation}
\epsilon_{\text{approx}}\leq \Vert\bm D^{-1}\Vert_{\infty}\sum\limits_{i=1}^n \left(\frac{1}{\sqrt{\mathrm{P}(\lambda_i)}}- \mathrm{P}_{\mbox{\scriptsize -1/2}}(\lambda_i) \right)^2 \leq n \Vert\bm D^{-1}\Vert_{\infty} \max_{x\in[a, b]} \left(\frac{1}{\sqrt{\mathrm{P}(x)}}- \mathrm{P}_{\mbox{\scriptsize -1/2}}(x) \right)^2
\label{err_approx}
\end{equation}
 
Hence, the approximation error on the matrices is upper-bounded by the overall error that arises from the polynomial approximation of $x \mapsto {1}/{\sqrt{\mathrm{P}(x)}}$ by its Chebyshev series. This last error can be made arbitrary small by truncating the polynomial series at a growing order. Therefore, the simulation algorithm is asymptotically exact given that asymptotically, the matrices $\bm D^{-1}\mathrm{P}(\bm S)^{-1/2}$ and $\bm D^{-1}\mathrm{P}_{\mbox{\scriptsize -1/2}}(\bm S)$ coincide.

%%%%%%%%%%%%%%%%%%%%%%%%%%%%%%%%%%%%%%%%%%%%%%%%%%%%%%%%%%

\subsection{Determination of the approximation precision}
\label{sec:test}

A practical question still needs to be answered : how to choose the order of the approximation polynomial? A criterion could be based on relation \eqref{err_approx}, by imposing an order high enough so that the approximation error $\epsilon_{\text{approx}}$ of the matrices is below a given tolerance. But then, this tolerance would also need to be chosen.

Given that the goal is actually to generate a random vector with the right covariance properties, the statistical properties of the random vectors produced by the proposed simulation algorithm are now studied. Criteria on the approximation error of the function $1/\sqrt{\mathrm{P}}$ (and therefore on the order of the polynomial $\mathrm{P}_{\mbox{\scriptsize -1/2}}$) are deduced to create vectors with statistical properties "close enough" to the targeted ones.

\subsubsection{Assessment of simulation validity through statistical tests}
\label{subsec:stat_test}
The aim is to simulate a Gaussian vector $\bm z$ with covariance matrix: 
$$\bm{\Sigma}=\bm{Q}^{-1}=\bm{D}^{-1}\mathrm{P}(\bm S)^{-1} \bm{D}^{-1}$$
But instead, the proposed simulation algorithm actually generates a Gaussian vector $\bm z_s$ with covariance matrix: $$\bm{\Sigma}_s=\bm{D}^{-1}\mathrm{P}_{\mbox{\scriptsize -1/2}}(\bm S)^2 \bm{D}^{-1}$$
Consider then a sample of $N$ independent zero-mean Gaussian vectors $\left(\bm z_s^{(1)}, \dots, \bm z_s^{(N)}\right)$ with covariance matrix $\bm{\Sigma}_s$. Let's consider the following  null hypothesis test:
\begin{center}
$H_0$ : $\left(\bm z_s^{(1)}, \dots, \bm z_s^{(N)}\right)$ is a sample of zero-mean Gaussian vectors with covariance matrix $\bm{\Sigma}$
\end{center} 

Obviously, the condition on the mean is satisfied by construction of this sample. Besides, by definition \cite{tong2012multivariate}, a random vector $\bm z$ is a Gaussian vector with covariance matrix $\bm\Sigma$ if and only if, for any $\bm v \in \mathbb{R}^n$, $\tr{\bm v}\bm z$ is a Gaussian variable with variance $\tr{\bm v}\bm{\Sigma}\bm v$. Therefore, hypothesis $H_0$ won't be rejected if $\forall\bm v\in\mathbb{R}^n$, the hypothesis $H_0^v$ defined by:
\begin{center}
$H_0^v$ : $\left(\tr{\bm v}\bm z_s^{(1)}, \dots, \tr{\bm v}\bm z_s^{(N)}\right)$ is a sample of zero-mean Gaussian variables with variance $\tr{\bm v}\bm{\Sigma}\tr{\bm v}$
\end{center}
is not rejected.

Two-sided chi-square tests for the variance \cite{snedecor1989} are considered. The results of these tests can actually be anticipated given that by definition, the sample $\left(\tr{\bm v}\bm z_s^{(1)}, \dots, \tr{\bm v}\bm z_s^{(N)}\right)$ has a known distribution: it is Gaussian with variance $\tr{\bm v} \bm{\Sigma}_s \bm v$. In particular, a criterion on the quality of the polynomial approximation such that \textit{for any} $\bm v \in \mathbb{R}^n$ the probability of rejecting hypothesis $H_0^v$ can be controlled is derived.

\begin{prop}
Let $[a, b]$ be an interval containing all the eigenvalues of $\bm S$. Then $\forall \gamma > 0$, there exists $\epsilon_{N,\gamma}>0$ such that:
\begin{equation}
\max\limits_{\lambda \in [a, b]} \left\vert \frac{1/\mathrm{P}(\lambda)-\mathrm{P}_{\mbox{\scriptsize -1/2}}(\lambda)^2}{\mathrm{P}_{\mbox{\scriptsize -1/2}}(\lambda)^2} \right| \leq \epsilon_{N,\gamma} \Rightarrow \forall \bm v\in\mathbb{R}^n, \quad R_\alpha(\bm v) \leq (1+\gamma)\alpha
\label{cond}
\end{equation}
where $R_\alpha(\bm v)$ is the probability of rejecting hypothesis $H_0^v$ in a chi-square test for the variance with significance $\alpha$.
\label{prop_approx_prec}
\end{prop}

\begin{proof}
See Appendix \ref{appen:proof}.
\end{proof}
Therefore, if \eqref{cond} is satisfied, then, for any $\bm v$, hypothesis $H_0^v$ is actually rejected (with significance $\alpha$) with a probability less than $(1+\gamma)\alpha$. This probability would have been equal to $\alpha$ if the samples were generated using the right covariance matrix. Therefore, the parameter $\gamma$ represents relative increase of the rejection probability due to the fact that the samples are generated using $\bm\Sigma_s$ instead of $\bm\Sigma$.

Tables \ref{tabPrec1} and \ref{tabPrec2} give typical values of the tolerance $\epsilon_{N,\gamma}$ for various sample sizes $N$ and thresholds $\gamma$. The significance is fixed at $\alpha = 0.05$ for table \ref{tabPrec1} and $\alpha = 0.01$ for table \ref{tabPrec2}. 

\begin{table}[h]
\centering
\begin{tabular}{|c||cccccc|}
\hline
\multirow{2}*{$\gamma$} & \multicolumn{6}{c|}{Sample size $N$} \\\cline{2-7}
 & 50 & 100 & 500 & 1000 & 5000 & 10000 \\ 
  \hline
0.1\% & \cellcolor{col1}6.40e-04 & \cellcolor{col1}6.20e-04 & \cellcolor{col1}5.40e-04 & \cellcolor{col1}4.80e-04 & \cellcolor{col1}3.00e-04 & \cellcolor{col1}2.40e-04 \\ 
  1\% & \cellcolor{col4}5.44e-03 & \cellcolor{col3}4.80e-03 & \cellcolor{col3}3.04e-03 & \cellcolor{col2}2.36e-03 & \cellcolor{col2}1.20e-03 & \cellcolor{col2}8.60e-04 \\ 
  5\% & \cellcolor{col6}1.89e-02 & \cellcolor{col6}1.51e-02 & \cellcolor{col4}8.06e-03 & \cellcolor{col4}5.94e-03 & \cellcolor{col2}2.82e-03 & \cellcolor{col2}2.02e-03 \\ 
  10\% & \cellcolor{col7}3.00e-02 & \cellcolor{col7}2.33e-02 & \cellcolor{col5}1.18e-02 & \cellcolor{col5}8.64e-03 & \cellcolor{col3}4.02e-03 & \cellcolor{col3}2.88e-03 \\ 
  20\% & \cellcolor{col8}4.59e-02 & \cellcolor{col7}3.48e-02 & \cellcolor{col6}1.71e-02 & \cellcolor{col5}1.24e-02 & \cellcolor{col4}5.74e-03 & \cellcolor{col3}4.08e-03 \\ 
  50\% & \cellcolor{col8}7.66e-02 & \cellcolor{col8}5.71e-02 & \cellcolor{col7}2.75e-02 & \cellcolor{col6}1.98e-02 & \cellcolor{col5}9.08e-03 & \cellcolor{col4}6.46e-03 \\ 
  100\% & \cellcolor{col8}1.10e-01 & \cellcolor{col8}8.12e-02 & \cellcolor{col8}3.89e-02 & \cellcolor{col7}2.80e-02 & \cellcolor{col6}1.28e-02 & \cellcolor{col5}9.10e-03 \\ 
   \hline
\end{tabular}
\caption{Values of the precision threshold $\epsilon_{N,\gamma}$ for different values of sample size $N$ and of degradation of the type I error $\gamma$. The significance of the test is $\alpha=0.05$}
\label{tabPrec1}

\vspace{4ex}

\begin{tabular}{|c||cccccc|}
\hline
\multirow{2}*{$\gamma$} & \multicolumn{6}{c|}{Sample size $N$} \\\cline{2-7}
 & 50 & 100 & 500 & 1000 & 5000 & 10000 \\ 
  \hline
0.1\% & \cellcolor{col1}4.00e-04 & \cellcolor{col1}4.00e-04 & \cellcolor{col1}3.60e-04 & \cellcolor{col1}3.20e-04 & \cellcolor{col1}2.20e-04 & \cellcolor{col1}1.80e-04 \\ 
  1\% & \cellcolor{col4}3.56e-03 & \cellcolor{col3}3.24e-03 & \cellcolor{col3}2.20e-03 & \cellcolor{col2}1.74e-03 & \cellcolor{col2}9.20e-04 & \cellcolor{col2}6.60e-04 \\ 
  5\% & \cellcolor{col6}1.33e-02 & \cellcolor{col6}1.09e-02 & \cellcolor{col4}6.06e-03 & \cellcolor{col4}4.52e-03 & \cellcolor{col2}2.18e-03 & \cellcolor{col2}1.56e-03 \\ 
  10\% & \cellcolor{col7}2.16e-02 & \cellcolor{col7}1.71e-02 & \cellcolor{col5}9.00e-03 & \cellcolor{col5}6.62e-03 & \cellcolor{col3}3.12e-03 & \cellcolor{col3}2.24e-03 \\ 
  20\% & \cellcolor{col8}3.36e-02 & \cellcolor{col7}2.59e-02 & \cellcolor{col6}1.31e-02 & \cellcolor{col5}9.54e-03 & \cellcolor{col4}4.44e-03 & \cellcolor{col3}3.18e-03 \\ 
  50\% & \cellcolor{col8}5.67e-02 & \cellcolor{col8}4.28e-02 & \cellcolor{col7}2.10e-02 & \cellcolor{col6}1.52e-02 & \cellcolor{col5}7.00e-03 & \cellcolor{col4}5.00e-03 \\ 
  100\% & \cellcolor{col8}8.11e-02 & \cellcolor{col8}6.07e-02 & \cellcolor{col8}2.94e-02 & \cellcolor{col7}2.12e-02 & \cellcolor{col6}9.76e-03 & \cellcolor{col5}6.96e-03 \\ 
   \hline
\end{tabular}
\caption{Values of the precision threshold $\epsilon_{N,\gamma}$ for different values of sample size $N$ and of degradation of the type I error $\gamma$. The significance of the test is $\alpha=0.01$}
\label{tabPrec2}
\end{table}

\subsubsection{Efficiency improvement}
\label{subsec:eff_impr}
The previous subsection provides a link between the order of the polynomial approximation of $1/\sqrt{P}$ and the validity of resulting simulations using the proposed algorithm. In practice, this order can be reduced in some cases.

Following the notations of section \ref{sec:alg}, let $(c_k)_{k\in\mathbb{N}}$ be the coefficients of the Chebyshev series of $1/\sqrt{P}$ over an interval $[a, b]$ containing all the eigenvalues of $\bm S$ and for $m\in\mathbb{N}^*$, let $\bm z^{(m)}$ be the vector defined by:
$$\bm z^{(m)}=\bm D^{-1}\left(\frac{1}{2} c_0 T_0^{[a, b]}(\bm S)\bm\varepsilon+\sum\limits_{k=1}^m c_k T_k^{[a, b]}(\bm S)\bm\varepsilon \right)$$
where $\bm\varepsilon$ is a vector with independent standard Gaussian values.
Then, for $m,l \in \mathbb{N}^*$:
\begin{align*}
\Vert \bm z^{(m+l)}-\bm z^{(m)} \Vert = \left\Vert \bm D^{-1}\sum\limits_{k=1}^l c_{m+k} T_k^{[a, b]}(\bm S)\bm\varepsilon \right\Vert
 \le \Vert \bm D^{-1}\Vert_\infty \sum\limits_{k=1}^l \vert c_{m+k} \vert \left\Vert T_k^{[a, b]}(\bm S)\bm\varepsilon \right\Vert 
\end{align*}
where $\Vert . \Vert$ denotes the Euclidean norm on $\mathbb{R}^n$.\\
Recall that, according to the properties of Rayleigh quotients (see Appendix \ref{appen:ray}), and using the symmetry of $T_k^{[a, b]}(\bm S)$:
$$\frac{\left\Vert T_k^{[a, b]}(\bm S)\bm\varepsilon \right\Vert^2 }{\Vert \bm\varepsilon\Vert^2}=\frac{\tr{\bm\varepsilon}T_k^{[a, b]}(\bm S)^2\bm\varepsilon  }{\tr{\bm\varepsilon}\bm\varepsilon}\leq \lambda_{\text{max}}\left(T_k^{[a, b]}(\bm S)^2 \right)$$
where $\lambda_{\text{max}}(.)$ denotes the largest eigenvalue of a matrix.
Notice then that, given that the Chebyshev polynomials are upper-bounded (in absolute value) by $1$, $\lambda_{\text{max}}\left(T_k^{[a, b]}(\bm S)^2 \right)\le 1$. Consequently,
$$\Vert \bm z^{(m+l)}-\bm z^{(m)} \Vert 
 \le \Vert \bm D^{-1} \Vert_\infty \sum\limits_{k=1}^l \vert c_{m+k} \vert \left\Vert\bm\varepsilon \right\Vert $$
Hence, $\Vert \bm D^{-1} \Vert_\infty\left\Vert\bm\varepsilon \right\Vert \sum\limits_{k=1}^l \vert c_{m+k} \vert \approx 0 \Rightarrow \bm z^{(m+l)}\approx\bm z^{(m)}$. This gives an additional criterion for the choice of the approximation order of the algorithm. After computing a value of order $L$ using the criterion based on statistical tests, its value can be decreased to $K$ as long as: $$\sum\limits_{k=K+1}^L \vert c_{k} \vert \le \frac{\eta}{\Vert\bm D^{-1} \Vert_\infty\Vert\bm\varepsilon\Vert}$$
for a fixed tolerance $\eta$ corresponding to the Euclidean distance between the vector computed using order $K$ and the one computed using order $L$. In particular, if $\eta$ is of the form $\eta=\sqrt{\epsilon n}$ (where $n$ is the size of the simulated vectors) and $\epsilon>0$, then in average, the square  of the components of the vector $(\bm z^{(L)} - \bm z ^{(K)})$  will be less than $\epsilon$.

\section{Application}
\label{sec:app}

\subsection{Simulation of stationary Matérn models}
\label{subec:sim_mat}

The Matérn model is a widely used covariance model in Geostatistics due to its great flexibility. For a lag distance $h\in\mathbb{R}_+$, its isotropic formulation is \cite{chiles1999geost}: 
\begin{equation*}
C(h)= \frac{\sigma^2}{2^{\nu-1}\Gamma(\nu)}\left(\frac{h}{\phi}\right)^{\nu}K_{\nu}\left(\frac{h}{\phi}\right)
\end{equation*}
where $\sigma^2>0$ is the marginal variance, $\phi >0$ is a scaling parameter, $\nu>0$ is a shape parameter and $K_\nu$ is the modified Bessel function of the second kind of order $\nu$. The parameter $\nu$ can be seen as a "smoothness" parameter as the underlying process is $\lfloor \nu\rfloor$-time mean-square differentiable.

In \cite{lindgren2011explicit}, Lindgren et al. define Markovian approximations GRF with Matérn covariance as numerical solutions of SPDE \eqref{spdeMatern} with parameters:
$$\kappa=1/\phi, \quad \alpha=\nu+d/2\in\mathbb{N}, \quad \tau=\sigma\kappa^{\nu}\sqrt{(4\pi)^{d/2}\Gamma(\nu+d/2)/\Gamma(\nu)}$$
using the finite element method with linear triangular elements. They show that the precision matrix of the weights of the finite element representation of the solution can be written as \eqref{Q_gen}, with:
\begin{equation}
\bm S =(1/\kappa^2)\bm C^{-1/2} \bm G \bm C^{-1/2}, \quad \bm D=(\kappa^\alpha/\tau)\bm C^{1/2}, \quad \mathrm{P}:x\mapsto (1+x)^{\alpha}
\label{app_def_mat}
\end{equation}
where:
\begin{equation}
\bm C = [\langle \psi_i, \psi_j \rangle], \quad \bm G = [\langle \nabla\psi_i, \nabla\psi_j \rangle], \quad \bm C^{-1/2}=(\bm C^{1/2})^{-1}
\label{el_C_G}
\end{equation}
and $\bm C^{1/2}$ is a square-root of $\bm C$. The matrices $\bm C$ and $\bm G$ are sparse. Following again \cite{lindgren2011explicit}, $\bm C$ is replaced by a diagonal matrix with entries $[\langle \psi_i, 1 \rangle]$, which yields a Markov approximation of the solution that has the same convergence rate as the full finite element method formulation.

In this subsection, simulations of the Matérn fields are generated using two methods (for comparison purposes): the Cholesky factorisation algorithm presented in Section \ref{subsec:matfac}, and the simulation algorithm proposed in this paper and presented in Section \ref{sec:alg}.

When applied to the simulation of random fields using the finite element method, the proposed algorithm can actually be interpreted as an image convolution algorithm. To see that, notice that when defined as in \eqref{app_def_mat}, the matrix $\bm S$  can be interpreted as what is referred to as a shift operator in the GSP theory \cite{shuman2013emerging}. Namely, when (linear) triangular elements are used, $\bm S$ is a sparse matrix whose entry $M_{ij}$ is non-zero if the nodes $i$ and $j$ of the triangulation are adjacent or if they coincide. Therefore, the non-zeros entries of a polynomial of degree $K$ of $\bm S$ correspond to nodes that are within a distance $K$ from each other in the triangulation (i.e. there exist a chain of at most $K$ adjacent nodes between them). 

Now, according to the workflow provided in Section \ref{sec:alg}, the output vector $\bm z$ of the simulation algorithm can be written as $\bm z=\bm D^{-1} \mathrm{P}_{\mbox{\scriptsize -1/2}}(\bm S)\bm\varepsilon$, where $\bm D$ is a diagonal matrix, $\bm \varepsilon$ is a vector indexed by the nodes of the triangulation and $\mathrm{P}_{\mbox{\scriptsize -1/2}}$ is a polynomial of degree $K$. So, the $i$-th entry of $\bm z$ is given by a linear combination of the entries of $\bm \varepsilon$ that correspond to nodes within a distance $K$ of $i$, and weighted by the entries of $i$-th row of $\bm D^{-1} \mathrm{P}_{\mbox{\scriptsize -1/2}}(\bm S)$. Hence, $\bm z$ can be seen as a convolution of the entries $\bm \varepsilon$. The bigger the degree of the polynomial, the bigger the size of the convolution kernel.

\subsubsection{Order of the polynomial approximation}
First, the effect of the order of the polynomial approximation on the resulting simulation is investigated. To do so, simulations of a Matérn field are generated on a 200x200 grid, with range 25, sill 1 and smoothness parameter 1, and with a growing order. In Figure \ref{fig_stat}, simulations obtained for degree values of 1, 5, 20 and 100 and the associated variogram (averaged over 50 simulations) are displayed. As a comparison, the same model simulated using the classical Cholesky factorisation algorithm is displayed in Figure \ref{fig_chol}.

\begin{figure}[b]
\vspace{-10pt}
     \centering
        \includegraphics[width=0.32\textwidth]{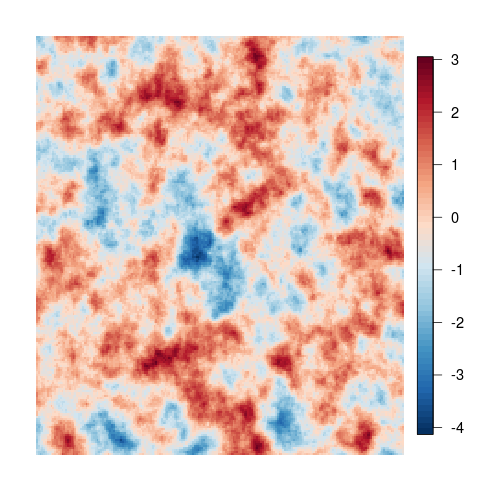}
        \includegraphics[width=0.32\textwidth]{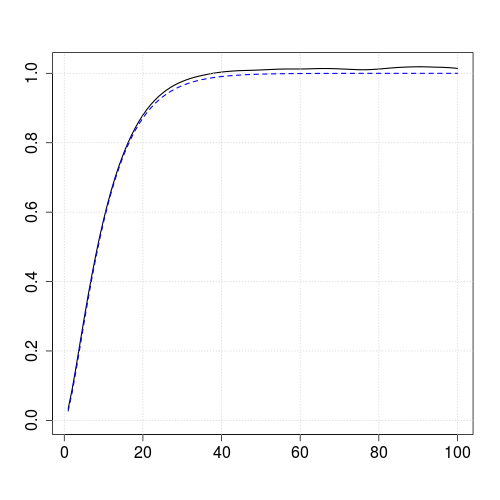}
        \vspace{-20pt}
 \caption{\textit{(Left)} Simulations  of a Matérn field on a 200x200 grid using Cholesky factorisation. \textit{(Right)} Mean variograms  over 50 simulations (solid line) and model (dotted line) : range=25, sill=1, smoothness=1.}
  	\label{fig_chol}
\end{figure}

\begin{figure}
\vspace{-10pt}
     \centering
    \begin{subfigure}[t]{\textwidth}
    	\centering
        \includegraphics[width=0.32\textwidth]{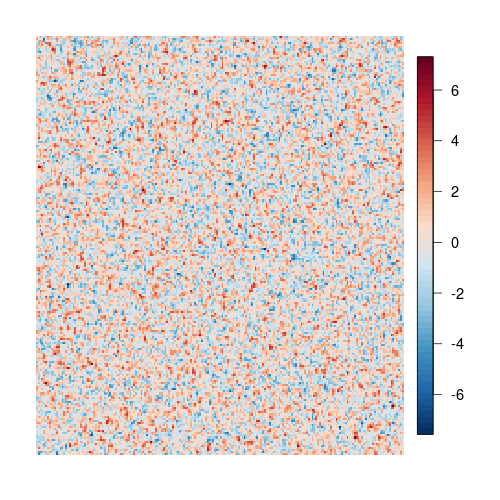}
        \includegraphics[width=0.32\textwidth]{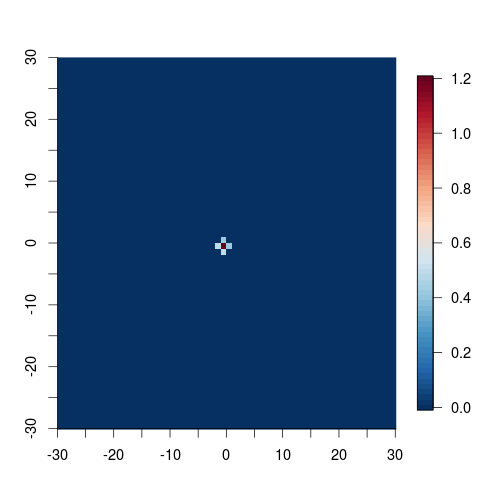}
        \includegraphics[width=0.32\textwidth]{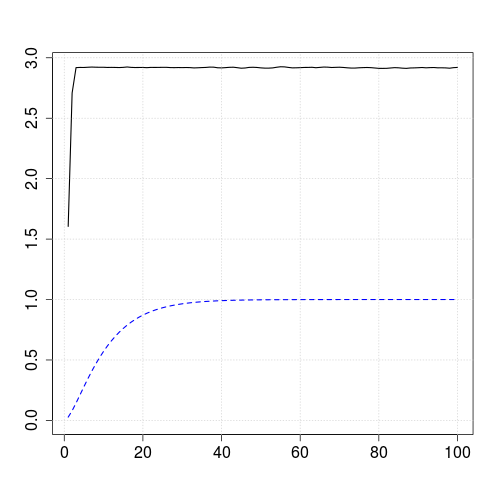}
        \vspace{-15pt}
        \caption{Order=1}
    \end{subfigure}
    \\
    \vspace{-10pt}
    \begin{subfigure}[t]{\textwidth}
    	\centering
        \includegraphics[width=0.32\textwidth]{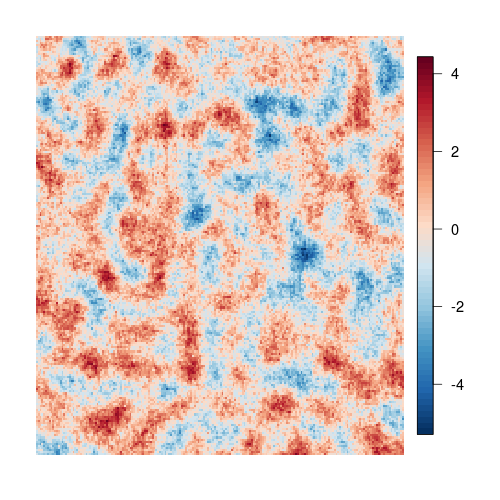}
         \includegraphics[width=0.32\textwidth]{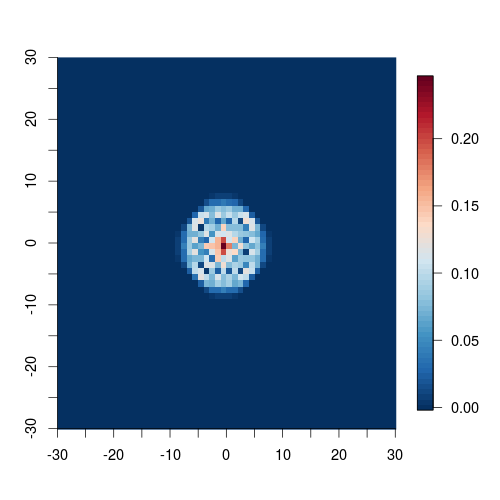}
        \includegraphics[width=0.32\textwidth]{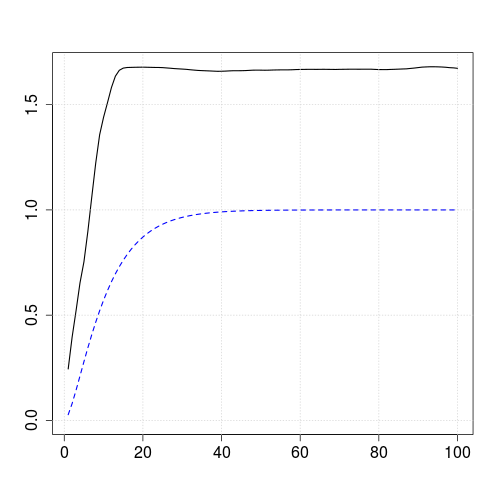}
        \vspace{-15pt}
        \caption{Order=10}
    \end{subfigure}
    \\
    \vspace{-10pt}
    \begin{subfigure}[t]{\textwidth}
    	\centering
        \includegraphics[width=0.32\textwidth]{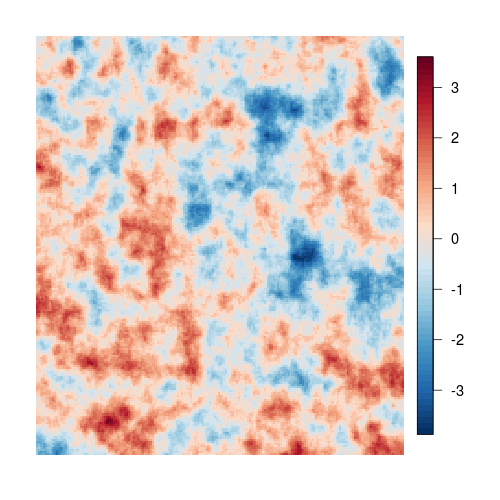}
    	\includegraphics[width=0.32\textwidth]{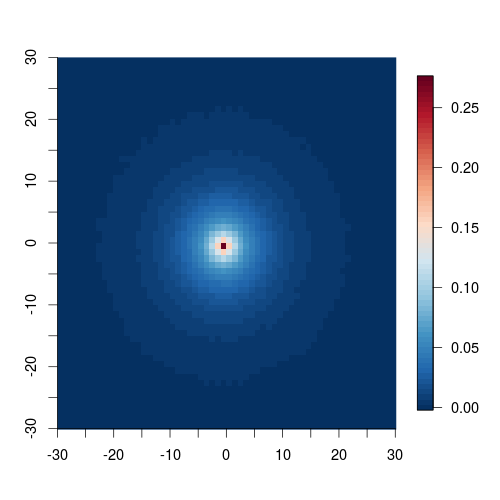}        
        \includegraphics[width=0.32\textwidth]{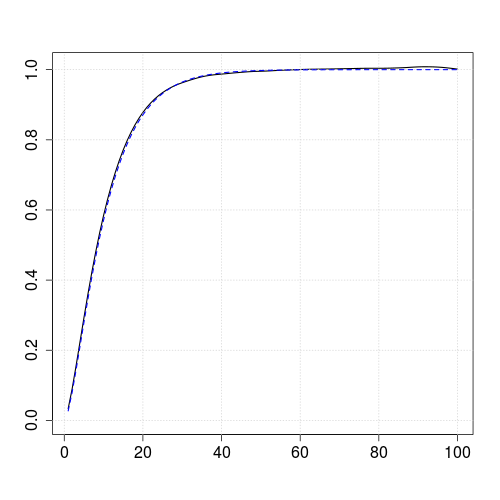}        
        \vspace{-15pt}
        \caption{Order=50}
    \end{subfigure}
        \\
    \vspace{-10pt}
    \begin{subfigure}[t]{\textwidth}
    	\centering
        \includegraphics[width=0.32\textwidth]{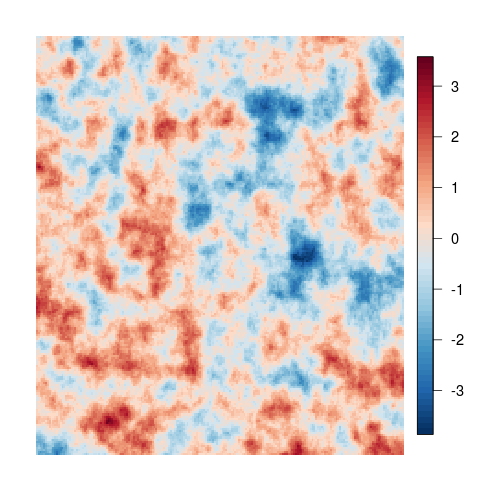}
    	\includegraphics[width=0.32\textwidth]{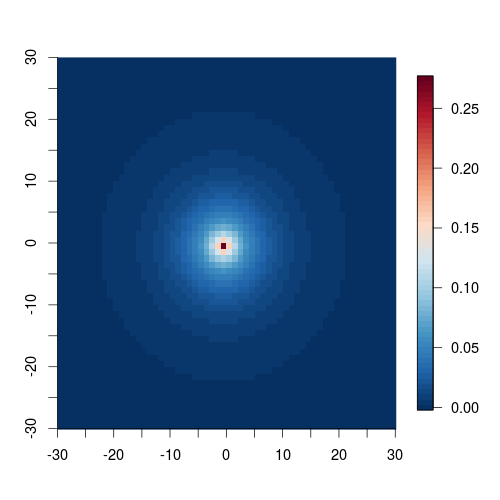}        
        \includegraphics[width=0.32\textwidth]{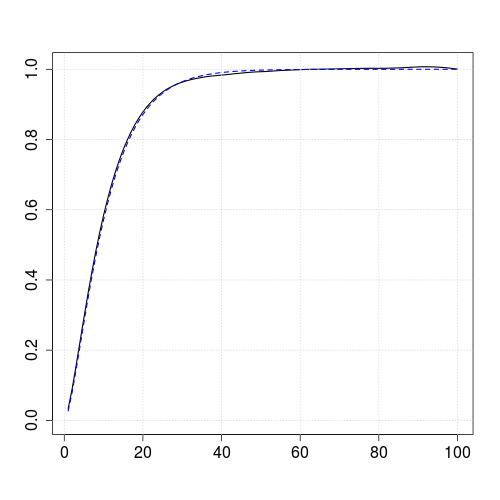}        
        \vspace{-15pt}
        \caption{Order=100}
    \end{subfigure}
    \vspace{-12pt}        
        \caption{\textit{(Left)} Simulations of a Matérn model with range 25, sill 1 and smoothness parameter 1 on a 200x200 grid using Chebyshev approximation with growing order. \textit{(Center)} Convolution kernels associated with the simulation. \textit{(Right)} Mean variograms  over 50 simulations (solid line) and model (dotted line).}
\label{fig_stat}
 \end{figure}
 
As noticed in Figure \ref{fig_stat}, increasing the order of the polynomial tends to add smoothness and structure to the simulation. This is expected from a convolution algorithm as the size of the kernel, which is directly linked to the order of the polynomial, grows (center images in Figure \ref{fig_stat}). Moreover, there seems to be a point from which adding more polynomials doesn't change the simulation. This observation is just a consequence of the result presented in Subsection \ref{subsec:eff_impr}.

Remark also that the variogram of the simulations in Figure \ref{fig_stat} tends to be respected as the order of the polynomial grows. This fact was predictable and is due to the fact that the proposed algorithm ensures that any linear combinations of the vectors generated by the algorithm have the right variance within a given tolerance (see subsection \ref{subsec:stat_test}). Consequently, this will ensure that the variogram  is respected given that its value at particular lag $h$ is just the variance of the difference between two particular entries of the simulated vector that correspond to nodes of the triangulation separated by an Euclidean distance of $h$.

\subsubsection{Influence of the model}

The influence of the covariance model parameters on the resulting approximations is now investigated. To do so, simulations of Matérn fields with different values of range and smoothness parameters are generated (cf. Figure \ref{fig_param_stat}). For each set of parameters, the order of approximation is set so that the probability of rejection on the statistical tests with significance $\alpha=0.05$ is equal to $(1+10\%)\alpha$, which corresponds to a threshold on the approximation error \eqref{cond} of 3.0e-02 (cf. Table \ref{tabPrec1}). Following the result of subsection \ref{subsec:eff_impr}, the effective order of approximation used to generate the simulation is reduced by considering a tolerance $\eta = \sqrt{10^{-4}n}$ where $n=200^2$ is the size of the simulated vector.

The order of approximation used for each simulation is reported in Figure \ref{fig_param_stat}. It can be noticed that increasing the range results in significantly higher orders of approximation to achieve the same accuracy, whereas the effect of the smoothness parameters seems more limited. This is just another consequence of the "convolution" nature of the algorithm, as explained at the beginning of the section. The bigger the range is, the bigger the size of the kernel used to generate the simulation from a white noise image should be as bigger "spots" must be created, and therefore the bigger the order of approximation is. On the other hand, the smoothness parameter mainly affects the smoothness of the kernel, not its size.

\begin{figure}[b!]
\vspace{-10pt}
     \centering
    \begin{subfigure}[t]{\textwidth}
    	\centering
        \includegraphics[width=0.32\textwidth]{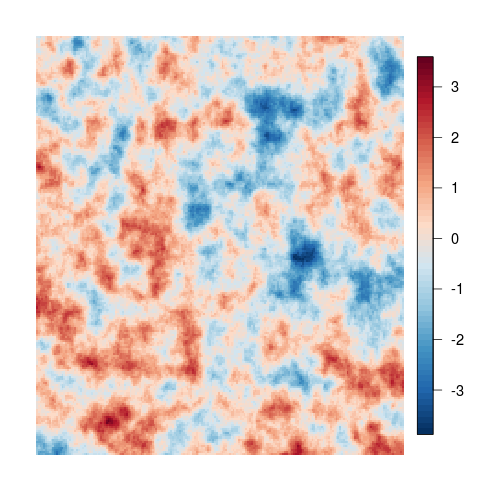}
    	\includegraphics[width=0.32\textwidth]{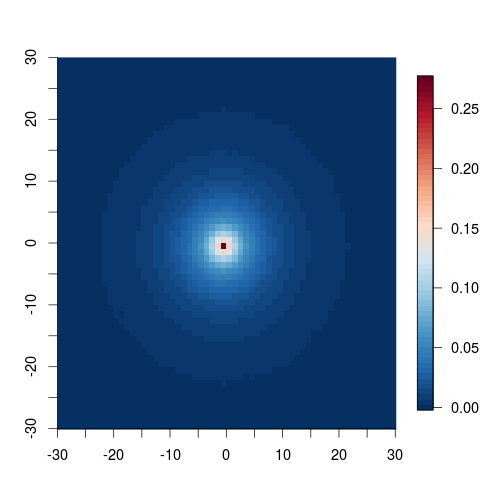}        
        \includegraphics[width=0.32\textwidth]{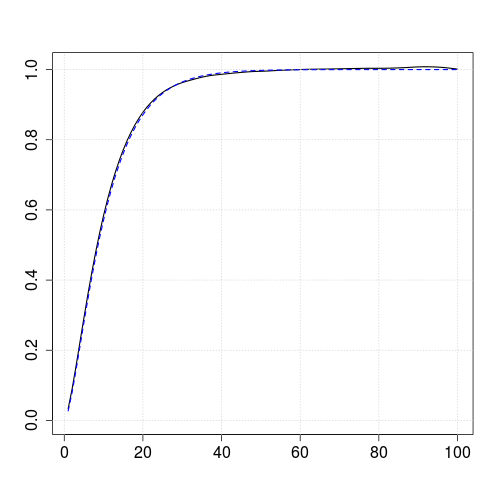}
        \vspace{-15pt}
        \caption{Range = 25, Smoothness = 1 : Order = 76 (Effective order = 52)}
    \end{subfigure}
    \\
    \vspace{-10pt}
    \begin{subfigure}[t]{\textwidth}
    	\centering
        \includegraphics[width=0.32\textwidth]{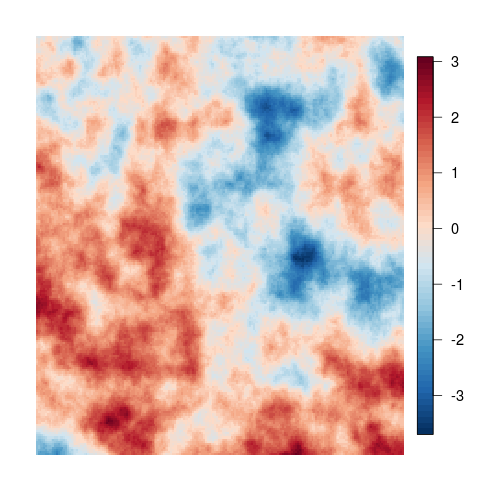}
    	 \includegraphics[width=0.32\textwidth]{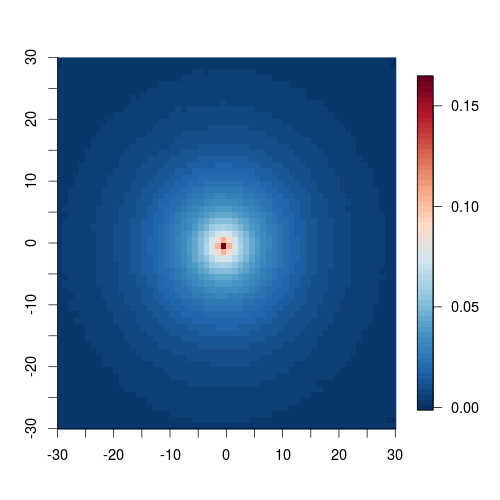}        
        \includegraphics[width=0.32\textwidth]{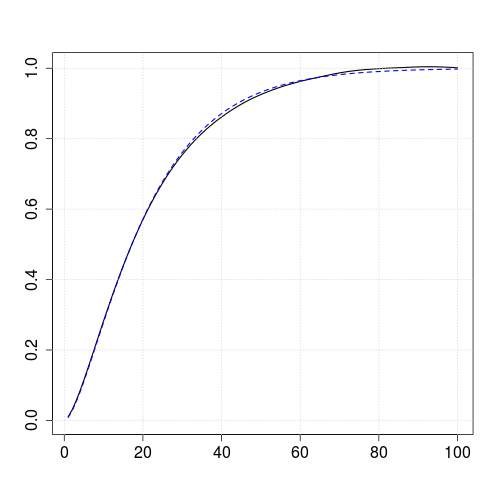}
        \vspace{-15pt}
        \caption{Range = 50, Smoothness = 1 : Order = 166 (Effective order = 102)}
    \end{subfigure}
    \\
    \vspace{-10pt}
    \begin{subfigure}[t]{\textwidth}
    	\centering
        \includegraphics[width=0.32\textwidth]{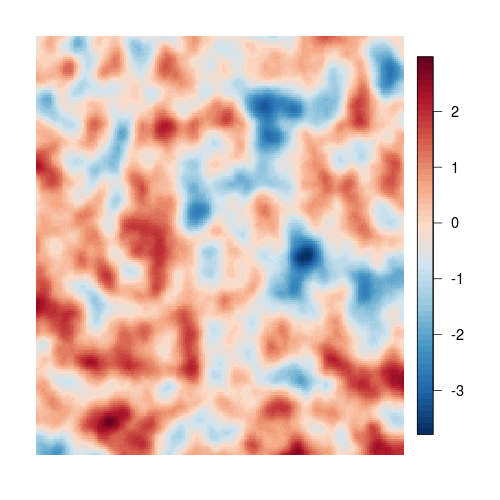}
    	\includegraphics[width=0.32\textwidth]{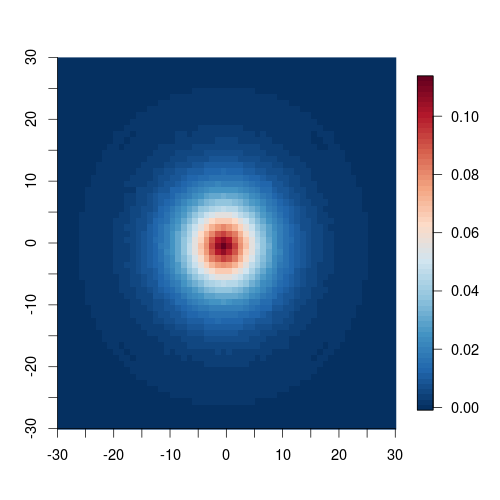}        
        \includegraphics[width=0.32\textwidth]{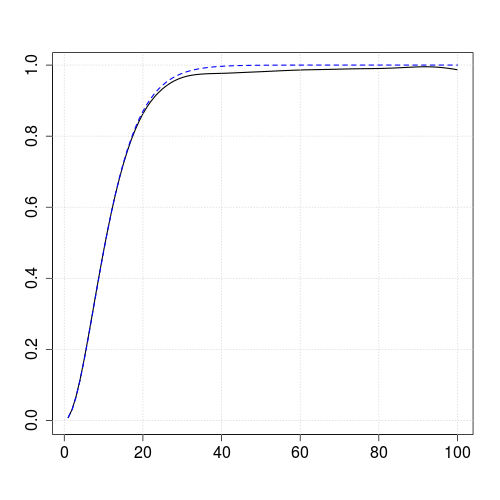}
        \vspace{-15pt}
        \caption{Range = 25, Smoothness = 3 : Order = 84 (Effective order = 40)}
    \end{subfigure}
        \vspace{-12pt}        
        \caption{\textit{(Left)} Simulations using Chebyshev approximation of a Matérn field on a 200x200 grid with various model parameters. \textit{(Center)} Convolution kernels associated with the simulation. \textit{(Right)} Mean variograms  over 50 simulations (solid line) and model (dotted line).}
    \label{fig_param_stat}
\end{figure}

\subsection{Simulation of non-stationary fields}

\begin{figure}[b!]
	\centering   
\vspace{-25pt}
	\begin{subfigure}[t]{0.8\textwidth}
		\includegraphics[width=\textwidth]{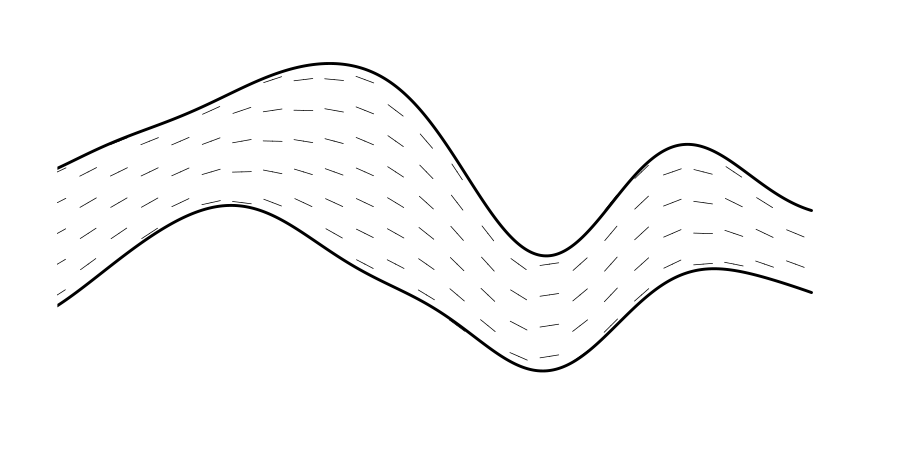}
        \vspace{-50pt}
        \caption{Anisotropy field : each line represents the local direction of the anisotropy. The anisotropy ratio is locally proportional to the thickness of the layer (maximum value : 1.5).}
    \end{subfigure}
	\hfill
	
	\vspace{-30pt}
	
    \begin{subfigure}[t]{0.8\textwidth}
		\includegraphics[width=\textwidth]{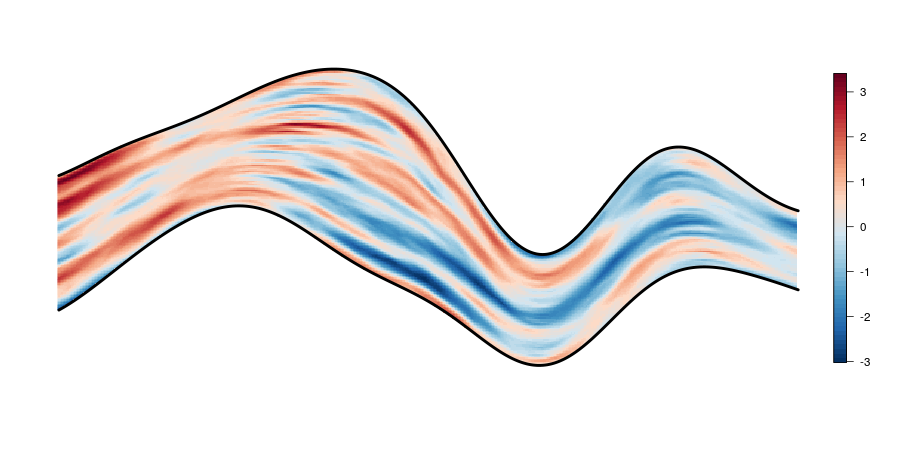}
        \vspace{-50pt}
        \caption{Simulation}
    \end{subfigure}
	\hfill
	
	\vspace{-20pt}
    \caption{Simulation using Chebyshev approximation of a non-stationary Matérn field with local anisotropy defined on a layer. Overall extension of the domain : 500x200. Model : range=150, sill=1.}
    \label{fig_nstat}
\end{figure}

Following the modeling approach of Fuglstad et al. \cite{fuglstad2015exploring}, an expression for the precision matrix of a finite element solution of the SPDE \eqref{spdeMatern} with locally varying coefficients is derived, resulting in a non-stationary field with local anisotropies. Specifically, let's consider a numerical solution of the SPDE:
\begin{equation*}
\kappa^2Z(x)-\mathop{\text{div}}\left(H(x)\nabla Z(x)\right)=\tau \mathcal{W}(x)
\end{equation*}
where $\kappa, \tau >0$, $\mathcal{W}$ is a spatial Gaussian white noise and $H$ is a field of positive definite matrices indexed by the space.
Using weak formulations of this SPDE, it can be shown that the precision matrix of the weights has the same expression as in the stationary case presented in the previous subsection (with $\alpha=2$) but where the matrix $\bm G$ is now defined by:
\begin{equation*}
\bm G = [\langle \nabla\psi_i, H\nabla\psi_j \rangle]
\end{equation*}
Notice that the matrix $\bm G$ contains all the "anisotropy" parameters of the model in the sense that, similarly as in \cite{fuglstad2015exploring}, its elements are constructed by locally accounting for the varying coefficients of the SPDE in expressions \eqref{el_C_G}. 

Just like in Subsection \ref{subec:sim_mat}, the Chebyshev approximation simulation algorithm to generate realizations of solutions of the SPDE. The results are displayed in Figure \ref{fig_nstat}.

\section{Conclusion}
In this paper, an algorithm for simulating Gaussian random vectors whose precision matrix can be expressed as a polynomial $\mathrm{P}$ of a sparse matrix $\bm S$ is proposed. This case arises in particular when simulation is carried out by solving numerically some SPDEs using finite element method. In doing so, random fields with a Matérn covariance can be generated. But this approach can also be used to simulate non-stationary random fields arising when SPDEs with varying coefficients are considered \cite{fuglstad2015exploring}.

The proposed algorithm is based on a computationally efficient polynomial approximation of a square-root of the covariance matrix, using a Chebyshev polynomial approximation of the function $1/\sqrt{\mathrm{P}}$, and can generate random vectors in $\mathcal{O}(Kn_{nz})$ operations, where $K$ is the order of approximation and $n_{nz}$ is the number of non-zeros of $\bm S$. The error of approximation on $1/\sqrt{\mathrm{P}}$ was proved to be directly linked to the statistical properties of the random vectors generated by the algorithm, namely the variance of linear combinations of their components, providing a criterion on the level of approximation for the simulated vectors to satisfy the right properties within a given tolerance. Due to its low computational complexity and memory requirements, this algorithm can be used to generate large random vectors, arising for instance from the discretization of the solutions of SPDEs in 3 dimensions, even with large smoothness parameters. This approach is implemented in the R package RGeostats \cite{Ren2018}. 

The proposed algorithm is a tool particularly well suited for the simulation of random fields represented as numerical solutions of SPDEs using finite element methods, and in particular Matérn fields and their extensions to manifolds, and oscillating and non-stationary covariances as presented in \cite{lindgren2011explicit}. Although only the case $\alpha\in\mathbb{N}$ (in SPDE \eqref{spdeMatern}) has been examined, the generalization to non-integer values of $\alpha$ is quite straightforward, as the resulting fields can be approximated by a linear combinations of fields with $\alpha\in\mathbb{N}$ \cite{lindgren2011explicit}. The main limitation seems rather to come from the finite element method itself which for small values of $\nu=\alpha-d/2$ requires thinner and thinner meshes and hence, bigger and bigger matrix sizes.

Finally, the proposed algorithm could be applied to the simulation of more general random fields than Matérn fields. Indeed, the  function $1/\sqrt{\mathrm{P}}$ approximated by the algorithm and used to generate the simulations, is proportional, for Matérn fields, to the actual square-root of the spectral density of the field. A way of generalizing the algorithm, which is under study and shows promising results, is to generate simulations of general Gaussian isotropic random fields using the same finite elements matrices as the one used here, but by replacing the approximation of the function $1/\sqrt{\mathrm{P}}$ by the approximation of the square-root of the spectral density of the targeted field.

%%%%%%%%%%%%%%%%%%%%%%%%%%%%%%%%%%%%%%%%%%%%%%%%%%%%%%%%%% Bibliography:
\nocite{*}
\bibliographystyle{siam}
\bibliography{biblio.bib}
\addcontentsline{toc}{section}{References}

%%%%%%%%%%%%%%%%%%%%%%%%%%%%%%%%%%%%%%%%%%%%%%%%%%%%%%%%%% Appendix:

\appendix

\section{Some elements of matrix analysis}
\label{appen:mat_ana}

Let $\bm n \in \mathbb{N}^*$. Let $\bm M=(M_{ij})_{i,j \in [\![1,n]\!]}$ be a real symmetric $n\times n$ matrix. $\bm M$ is diagonalizable in an orthonormal basis. Denote $\lambda_{\min}=\lambda_1 \le \dots \le \lambda_n=\lambda_{\max}$ its eigenvalues, and $(\bm v^{(1)}, \dots, \bm v^{(n)})$ the corresponding eigenvectors (forming the orthonormal basis).

\subsection{Rayleigh quotient}
\label{appen:ray}

\begin{definition}
The \textbf{Rayleigh quotient} $R(\bm M, \bm x)$ associated with $\bm M$ and $\bm x \in (\mathbb{R}^n)^*$ is the ratio:
$$R(\bm M, \bm x)=\frac{\tr{\bm x}\bm M \bm x}{\tr{\bm x}\bm x}=\tr{\left(\frac{\bm x}{\Vert \bm x\Vert}\right)}\bm M \left(\frac{\bm x}{\Vert \bm x\Vert}\right)$$
\end{definition}

\begin{prop}
$\forall \bm x \in (\mathbb{R}^n)^*, \quad \lambda_{\min}\le R(\bm M, \bm x)\le \lambda_{\max}$
\end{prop}
\begin{proof}
Simply notice that $\bm x$ can be decomposed in the orthonormal basis $(\bm v^{(1)}, \dots, \bm v^{(n)})$ as:
$\bm x = \sum_{i=1}^n \alpha_i \bm v^{(i)}$
for some $(\alpha_1, \dots, \alpha_n)^T\in (\mathbb{R}^n)^*$. Then, 
$$R(\bm M, \bm x)= \frac{\sum_{i=1}^n \lambda_i \alpha_i^2}{ \sum_{j=1}^n \alpha_j^2}=\sum_{i=1}^n \lambda_i \frac{\alpha_i^2}{ \sum_{j=1}^n \alpha_j^2}$$
which is just a weighted sum of the eigenvalues with positive weights.
\end{proof}

\subsection{Eigenvalue bounds}
\label{appen:eig_bound}

%Suppose now that $\bm M$ is a positive semi-definite matrix. Then, $0\le \lambda_{\min}=\lambda_1 \le \dots \le \lambda_n=\lambda_{\max}$.
\begin{prop}
$\forall i\in[\![1,n]\!], \quad \vert \lambda_{i}\vert \le \sqrt{\trace\left(\bm M^2 \right)}$. \\ All the eigenvalues of $\bm M$ are therefore included in the interval $[-\sqrt{\trace\left(\bm M^2 \right)}, \sqrt{\trace\left(\bm M^2 \right)}]$.
\end{prop}
\begin{proof}
Notice that $\trace\left(\bm M^2 \right)=\sum_{j=1}^n\lambda_j^2$. Take $i\in[\![1,n]\!]$. If $\lambda_i= 0$, the statement is true. Otherwise, $\trace\left(\bm M^2 \right)=\lambda_{i}^2\left(1+\sum_{j\ne i}(\lambda_{j}/\lambda_{i})^2\right)\ge \lambda_{i}^2$ and therefore the statement is also true.
\end{proof}

\begin{theorem}
\emph{Gerschgorin circle theorem \cite{gerschgorin1931uber}.} \\
Any eigenvalue $\lambda$ of $\bm M$ satisfies:
\begin{equation*}
\lambda \in \mathop{\bigcup}\limits_{i\in [\![1,n]\!]} [M_{ii} - r_i, M_{ii} + r_i], \quad r_i=\sum\limits_{j\neq i} |M_{ij}|
\end{equation*}
\label{gers}
\end{theorem}
\begin{proof} Take $\lambda$ an eigenvalue of $\bm M$ and $\bm v=(v_1, \dots, v_n)^T$ an associated eigenvector. Let $i_0$ be the index of the component of $\bm v$ with the largest magnitude: $\forall j, \vert v_j\vert \le \vert v_{i_0}\vert$. \\
Then, given that $\bm M \bm v = \lambda \bm v$:
$\lambda v_{i_0}=\sum_{j=1}^n M_{i_0j}v_j= M_{i_0i_0}v_{i_0}+\sum_{j\neq i_0} M_{i_0j}v_j $, which gives:
\begin{nscenter}
$ (\lambda - M_{i_0i_0})v_{i_0}=\sum_{j\neq i_0} M_{i_0j}v_j $
\end{nscenter}
So:
\begin{nscenter}
$\vert \lambda - M_{i_0i_0} \vert \vert v_{i_0} \vert = \left\vert\sum_{j\neq i_0} M_{i_0j}v_j \right\vert \le  \sum_{j\neq i_0} \vert M_{i_0j}\vert \vert v_j \vert$
\end{nscenter}
And finally,
\begin{nscenter}
$\vert \lambda - M_{i_0i_0} \vert \le  \sum_{j\neq i_0} \vert M_{i_0j}\vert \frac{ \vert v_j \vert}{\vert v_{i_0} \vert }\le r_{i_0}$
\end{nscenter}

\end{proof}

\section{Proof of Proposition \ref{prop_approx_prec}}
\label{appen:proof}
\begin{proof}

Take $\bm v\in (\mathbb{R}^n)^*$. Let's perform a chi-squared test for the variance on hypothesis $H_0^v$. The statistic $t(\bm v)$ of this test is: 
$$t(\bm v)=(N-1)\frac{S^2(\bm v)}{\tr{\bm v}\bm{\Sigma}\bm v}$$
where $S^2(\bm v)$ is the (unbiased) sample variance defined as:
$$S^2(\bm v)=\frac{1}{N-1}\sum\limits_{i=1}^N \left(\tr{\bm v}\bm z_s^{(i)}-m(\bm v)\right)^2, \quad m(\bm v)=\frac{1}{N}\sum\limits_{i=1}^N \tr{\bm v}\bm z_s^{(i)}$$
$H_0^v$ will not be rejected with significance $\alpha$ if $t(\bm v)$ satisfies : 
$$\chi^2_{\frac{\alpha}{2},N-1}\leq t(\bm v) \leq \chi^2_{1-\frac{\alpha}{2},N-1}$$
where $\chi^2_{p,N-1}$ is the $p$-th quantile of the chi-squared distribution with $N-1$ degrees of freedom (denoted $\chi^2(N-1)$). 

In particular, the probability $R_{\alpha}(\bm v)$ that $H_0^v$ is rejected with significance $\alpha$ is given by:
\begin{align*}
R_{\alpha}(\bm v)&=1-P\left(\chi^2_{\frac{\alpha}{2},N-1}\leq t(\bm v) \leq \chi^2_{1-\frac{\alpha}{2},N-1}\right) \\
&=1-P\left(\frac{\tr{\bm v}\bm{\Sigma}\bm v}{\tr{\bm v}\bm{\Sigma}_s\bm v}\chi^2_{\frac{\alpha}{2},N-1}\leq t_{s}(\bm v) \leq \frac{\tr{\bm v}\bm{\Sigma}\bm v}{\tr{\bm v}\bm{\Sigma}_s\bm v}\chi^2_{1-\frac{\alpha}{2},N-1}\right)
\end{align*}
where $t_{s}(\bm v)$ is the statistic defined by: 
$$t_{s}(\bm v)=\frac{\tr{\bm v}\bm{\Sigma}\bm v}{\tr{\bm v}\bm{\Sigma}_s\bm v}t(\bm v)=(N-1)\frac{S^2}{\tr{\bm v}\bm{\Sigma}_s\bm v}$$
By definition, the sample $\left(\tr{\bm v}\bm z_s^{(1)}, \dots, \tr{\bm v}\bm z_s^{(N)}\right)$ is Gaussian with variance $\tr{\bm v} \bm{\Sigma}_s \bm v$. Hence, $t_{s}(\bm v)$ follows a $\chi^2(N-1)$ distribution.
So, if $F_{\chi^2(N-1)}$ denotes the cumulative distribution function of the $\chi^2(N-1)$ distribution: 
\begin{equation*}
R_{\alpha}(\bm v)=R_{\alpha}(X)=1-\left[F_{\chi^2(N-1)}\left(\chi^2_{1-\frac{\alpha}{2},N-1}X\right) - F_{\chi^2(N-1)}\left(\chi^2_{\frac{\alpha}{2},N-1}X \right)\right], \quad X:=\frac{\tr{\bm v}\bm{\Sigma}\bm v}{\tr{\bm v}\bm{\Sigma}_s\bm v}
\end{equation*}

If the variance of the simulated sample were to be equal to the true variance (i.e. $X=1$), the probability $R_{\alpha}$ of rejecting the test  would be $\alpha$, the type I error of the test. But here, this error depends on the ratio $X$. Due to the non-symmetry of the chi-squared distribution, it can even be smaller than the significance level $\alpha$ of the test. However, when the size of the sample increases, the $\chi^2$ distribution tends to regain symmetry and the minimum of the function $R_{\alpha}(X)$ tends to be achieved at $X=1$ for a value $\alpha$.

Denote $\gamma_\alpha(X)$ the ratio:
$$\gamma_\alpha(X)=\frac{R_\alpha(X)-\alpha}{\alpha}$$
This ratio measures the degradation of the type I error of test, due to the fact that the simulated sample has variance $\tr{\bm v}\bm{\Sigma}_s\bm v$ instead of $\tr{\bm v}\bm{\Sigma}\bm v$.
A condition on the value of $X$ so that $\gamma_\alpha(X) \leq \gamma$ for some fixed proportion $\gamma\geq 0$ can be expressed for some $\epsilon_{N,\gamma}>0$ as follows : 
\begin{equation}
\vert X-1  \vert \leq \epsilon_{N,\gamma} \Rightarrow \gamma_\alpha(X)\leq \gamma
\label{condX}
\end{equation}
For fixed values of $\gamma$, $\alpha$ and $N$ the value of $\epsilon_{N,\gamma}$ is determined numerically by finding the two roots of the equation $\gamma_\alpha(X)= \gamma$.

The condition on $X$ in \eqref{condX} can be written in terms of variances:
\begin{equation*}
\left\vert \frac{\tr{\bm v}(\bm{\Sigma}-\bm{\Sigma}_s)\bm v}{\tr{\bm v}\bm{\Sigma}_s\bm v} \right\vert \leq \epsilon_{N,\gamma}
\end{equation*}
which in turn, by denoting $\bm{\Sigma}_s ^{1/2}=\mathrm{P}_{\mbox{\scriptsize -1/2}}(\bm S)\bm D^{-1} $, can be expressed as: 
\begin{equation*}
\left\vert \tr{\left(\frac{\bm{\Sigma}_s ^{1/2}\bm v}{\Vert \bm{\Sigma}_s ^{1/2}\bm v \Vert }\right)}\mathrm{P}_{\mbox{\scriptsize -1/2}}(\bm S)^{-1}(\mathrm{P}(\bm S)^{-1}-\mathrm{P}_{\mbox{\scriptsize -1/2}}(\bm S)^2)\mathrm{P}_{\mbox{\scriptsize -1/2}}(\bm S)^{-1}\left(\frac{\bm{\Sigma}_s ^{1/2}\bm v}{\Vert \bm{\Sigma}_s ^{1/2}\bm v \Vert } \right)\right\vert \leq \epsilon_{N,\gamma}
\end{equation*}
Noticing the Rayleigh quotient (cf. Appendix \ref{appen:ray}) on the left side of the inequality, this relation can be satisfied \textit{for any $\bm v\in\mathbb{R}^n$} by simply imposing:
$$\lambda_{\max\text{ magn}}\left(\mathrm{P}_{\mbox{\scriptsize -1/2}}(\bm S)^{-1}(\mathrm{P}(\bm S)^{-1}-\mathrm{P}_{\mbox{\scriptsize -1/2}}(\bm S)^2)\mathrm{P}_{\mbox{\scriptsize -1/2}}(\bm S)^{-1}\right) \leq \epsilon_{N,\gamma}$$
where $\lambda_{\max \text{ magn}}(.)$ denotes the eigenvalue with the greatest magnitude of a matrix.
In turn, if $[a, b]$ denotes an interval containing all the eigenvalues of $\bm S$, this last condition can be satisfied by imposing : 
\begin{equation*}
\max\limits_{\lambda \in [a, b]} \left\vert \frac{1/\mathrm{P}(\lambda)-\mathrm{P}_{\mbox{\scriptsize -1/2}}(\lambda)^2}{\mathrm{P}_{\mbox{\scriptsize -1/2}}(\lambda)^2} \right| \leq \epsilon_{N,\gamma} 
\end{equation*}
\end{proof}

\end{document}